\documentclass[onecolumn]{revtex4-2}
\usepackage[
	bookmarks  = false,
	colorlinks = true,
	linkcolor  = blue,
	urlcolor   = blue,
	citecolor  = blue]{hyperref}
\usepackage{amsmath,amssymb,amsthm,graphicx}
\usepackage{braket,bm,bbm}
\usepackage[T1]{fontenc}
\newtheorem{theorem}{Theorem}
\newtheorem{lemma}{Lemma}
\newtheorem{corollary}{Corollary}
\newtheorem{definition}{Definition}

\newtheorem{example}{Example}
\makeatletter
\def\thmhead@plain#1#2#3{%
  \thmname{#1}\thmnumber{\@ifnotempty{#1}{ }\@upn{#2}}%
  \thmnote{ {\the\thm@notefont#3}}}
\let\thmhead\thmhead@plain
\makeatother
\DeclareMathOperator{\Tr}{Tr}
\DeclareMathOperator{\PT}{PT}
\begin{document}
%%%%%%%%%%%%%%%%%%%%%%
%       Title        %
%%%%%%%%%%%%%%%%%%%%%%
\title{Quantum data-hiding scheme using orthogonal separable states}
\author{Donghoon Ha}
\affiliation{Department of Applied Mathematics and Institute of Natural Sciences, Kyung Hee University, Yongin 17104, Republic of Korea}
\author{Jeong San Kim}
\email{freddie1@khu.ac.kr}
\affiliation{Department of Applied Mathematics and Institute of Natural Sciences, Kyung Hee University, Yongin 17104, Republic of Korea}
%%%%%%%%%%%%%%%%%%%%%%
%      Abstract      %
%%%%%%%%%%%%%%%%%%%%%%
\begin{abstract}
We consider bipartite quantum state discrimination and present a quantum data-hiding scheme utilizing an orthogonal separable state ensemble. Using a bound on local minimum-error discrimination, we provide a sufficient condition for the separable state ensemble to be used in constructing a quantum data-hiding scheme. Our results are illustrated with various examples in bipartite quantum systems. As our scheme employs separable states of low-dimensional quantum systems, it becomes more feasible for practical implementation.
\end{abstract}
%%%%%%%%%%%%%%%%%%%%%%
\maketitle
%%%%%%%%%%%%%%%%%%%%%%

%%%%%%%%%%%%%%%%%%%%%%
%      Section       %
%%%%%%%%%%%%%%%%%%%%%%
\section{Introduction}\label{sec:intro}
%%%%%%%%%%%%%%%%%%%%%%
In the realm of information theory, data hiding refers to a communication protocol that encodes specific data and distributes it among multiple users. Each user cannot decrypt the hidden data with only their received information, requiring cooperation with others to access it. Classical data hiding involves dividing the data into classical information and distributing it to users\cite{sham1979}. The hidden data can only be recovered when a sufficient number of these divided pieces are combined through classical communication. Thus, blocking classical communication is essential for hiding data in classical schemes.
%%%%%%%%%%%%%%%%%%%%%%

%%%%%%%%%%%%%%%%%%%%%%
Quantum data hiding encodes classical data into multipartite quantum states, providing a fundamentally stronger level of concealment compared with classical methods\cite{divi2002,divi2003,lami2018}. While the data remains hidden from each user's local operations, quantum data hiding can maintain concealment even if classical communication between users is allowed. Specifically, the probability of successfully recovering the hidden data using only \emph{local operations and classical communication}(LOCC) is nearly equivalent to random guessing. Successful recovery requires global measurements that involve quantum channels, shared entanglement, or direct interactions.
%%%%%%%%%%%%%%%%%%%%%%

%%%%%%%%%%%%%%%%%%%%%%
Since the initial use of Bell-state ensembles to hide one classical bit, there have been several quantum data-hiding protocols proposed not only in bipartite quantum systems, but also in multipartite quantum systems (involving multiple users)\cite{terh2001,mori2013,ha20241,ha20251}. 
These protocols typically employ orthogonal state ensembles to ensure perfect data recovery through appropriate global measurements. 
The use of entangled states in these ensembles may seem the reason of the nonlocality making the data hidden from LOCC measurements. 
%%%%%%%%%%%%%%%%%%%%%%

%%%%%%%%%%%%%%%%%%%%%%
Interestingly, it has also been demonstrated that quantum state ensembles containing only separable states can also be used to construct data-hiding schemes\cite{egge2002}. However, the separable states used in the protocol are non-orthogonal, and data recovery through appropriate global measurements is only asymptotically perfect with respect to the dimensionality.
This naturally raises the question of whether quantum data hiding can be constructed using only orthogonal separable states.
%%%%%%%%%%%%%%%%%%%%%%

%%%%%%%%%%%%%%%%%%%%%%
Here, we present a quantum data-hiding scheme utilizing an orthogonal separable state ensemble. The orthogonality of these states ensures perfect data recovery without reliance on asymptotic methods. While our scheme does asymptotically conceal classical data from LOCC measurements, it employs separable states of low-dimensional quantum systems repeatedly, making it more feasible for practical implementation.
%%%%%%%%%%%%%%%%%%%%%%

%%%%%%%%%%%%%%%%%%%%%%
This paper is organized as follows:
In Sec.~\ref{sec:pre}, we first consider the discrimination of two bipartite quantum states and provide a bound on the optimal local discrimination when the state ensemble possesses specific properties related to partial transposition. 
In Sec.~\ref{sec:mfe}, we introduce a quantum data-hiding scheme utilizing an orthogonal state ensemble that exhibits these same properties associated with partial transposition.
In Sec.~\ref{sec:ex}, we establish a sufficient condition for an ensemble of two bipartite quantum states to be used in constructing a one-bit hiding scheme. 
Our results are illustrated with examples of orthogonal separable state ensembles. 
In Sec.~\ref{sec:dis}, we provide a summary of our findings and explore potential avenues for future research.
%%%%%%%%%%%%%%%%%%%%%%

%%%%%%%%%%%%%%%%%%%%%%
%      Section       %
%%%%%%%%%%%%%%%%%%%%%%
\section{Minimum-error discrimination between two bipartite quantum states}\label{sec:pre}
%%%%%%%%%%%%%%%%%%%%%%
In bipartite quantum systems held by two players, Alice and Bob, a state is described by a density operator $\rho$, that is, a positive-semidefinite operator $\rho\succeq0$ with $\Tr\rho=1$, acting on a bipartite complex Hilbert space $\mathcal{H}=\mathbb{C}^{d_{\sf A}}\otimes\mathbb{C}^{d_{\sf B}}$.
A measurement is expressed by a positive operator-valued measure $\{M_{i}\}_{i}$, that is, a set of positive-semidefinite operators $M_{i}\succeq0$ acting on $\mathcal{H}$ and satisfying $\sum_{i}M_{i}=\mathbbm{1}$, where $\mathbbm{1}$ is the identity operator acting on $\mathcal{H}$.
When $\{M_{i}\}_{i}$ is performed on a quantum system prepared with $\rho$, the probability of obtaining the measurement outcome corresponding to $M_{j}$ is $\Tr(\rho M_{j})$.
%%%%%%%%%%%%%%%%%%%%%%

%%%%%%%%%%%%%%%%%%%%%%
Let us consider the situation of discriminating two bipartite quantum states $\rho_{0}$ and $\rho_{1}$ prepared with probabilities $\eta_{0}$ and $\eta_{1}$, respectively. We use the term \emph{state discrimination of an ensemble} 
\begin{equation}\label{eq:ens}
\mathcal{E}=\{\eta_{0},\rho_{0};\eta_{1},\rho_{1}\}.
\end{equation}
If we use a measurement 
\begin{equation}\label{eq:meas}
\mathcal{M}=\{M_{0},M_{1}\}
\end{equation} 
to discriminate the states of $\mathcal{E}$, 
our decision rule is the following; for each $i=0,1$, the prepared state is guessed to be $\rho_{i}$ when $M_{i}$ is detected.
The \emph{minimum-error discrimination} of $\mathcal{E}$ is to achieve the maximum average probability of correctly guessing the prepared state from $\mathcal{E}$, that is,
\begin{equation}\label{eq:pge}
p_{\sf G}(\mathcal{E})=\max_{\mathcal{M}}\sum_{i=0}^{1}\eta_{i}\Tr(\rho_{i}M_{i}),
\end{equation}
where the maximum is taken over all possible measurements\cite{hels1969}. 
%%%%%%%%%%%%%%%%%%%%%%

%%%%%%%%%%%%%%%%%%%%%%
%     Definition     %
%%%%%%%%%%%%%%%%%%%%%%
\begin{definition}\label{def:ppt}
A Hermitian operator $E$ on $\mathcal{H}$ is called \emph{positive partial transpose}(PPT) if the partially transposed operator is positive semidefinite; that is,
\begin{equation}\label{eq:ppt}
E^{\PT}\succeq0
\end{equation}
where the superscript $\PT$ is to indicate the partial transposition\cite{pere1996,pptp}.
\end{definition}
%%%%%%%%%%%%%%%%%%%%%%

%%%%%%%%%%%%%%%%%%%%%%
A measurement $\{M_{i}\}_{i}$ is called a \emph{PPT measurement} if $M_{i}$ is PPT for all $i$, and a measurement is called a LOCC measurement if it can be realized by LOCC.
When the available measurements are limited PPT measurements, we denote the maximum success probability by
\begin{equation}\label{eq:pte}
p_{\sf PPT}(\mathcal{E})=\max_{\mathsf{PPT}\,\mathcal{M}}\sum_{i=0}^{1}\eta_{i}\Tr(\rho_{i}M_{i}).
\end{equation}
Similarly, we denote 
\begin{equation}\label{eq:ple}
p_{\sf L}(\mathcal{E})=\max_{\mathsf{LOCC}\,\mathcal{M}}\sum_{i=0}^{1}\eta_{i}\Tr(\rho_{i}M_{i}),
\end{equation} 
where the maximum is taken over all possible LOCC measurements.
From the definitions of $p_{\sf G}(\mathcal{E})$, $p_{\sf PPT}(\mathcal{E})$ and $p_{\sf L}(\mathcal{E})$, we have
\begin{equation}\label{eq:pgplr}
p_{\sf L}(\mathcal{E})\leqslant
p_{\sf PPT}(\mathcal{E})\leqslant
p_{\sf G}(\mathcal{E})
\end{equation}
where the first inequality is from the fact that every LOCC measurement is a PPT measurement\cite{chit2014}.
Since guessing the prepared state as the state with the greatest probability is obviously a LOCC measurement, we also note that
\begin{equation}\label{eq:lptm}
\tfrac{1}{2}\leqslant\max\{\eta_{0},\eta_{1}\}\leqslant p_{\sf L}(\mathcal{E}).
\end{equation}
%%%%%%%%%%%%%%%%%%%%%%

%%%%%%%%%%%%%%%%%%%%%%
For a Hermitian operator $H$ acting on the Hilbert space $\mathcal{H}$, we define $q(H)$ as the quantity
\begin{equation}\label{eq:dqe}
q(H)=\tfrac{1}{2}+\Tr|H|
\end{equation}
where $|H|$ is the positive square root of $H^{\dagger}H$, that is,
\begin{equation}
|H|=\sqrt{H^{\dagger}H}.
\end{equation}
The following theorem provides a sufficient condition for $q(H)$ to be an upper bound of $p_{\sf PPT}(\mathcal{E})$.
%%%%%%%%%%%%%%%%%%%%%%
%      Theorem       %
%%%%%%%%%%%%%%%%%%%%%%
\begin{theorem}\label{thm:upbd}
For an ensemble $\mathcal{E}=\{\eta_{0},\rho_{0};\eta_{1},\rho_{1}\}$ and
\begin{equation}\label{eq:spc}
\Lambda_{\mathcal{E}}:=\eta_{0}\rho_{0}-\eta_{1}\rho_{1},
\end{equation} 
if there exists a Hermitian operator $H$ satisfying
\begin{equation}\label{eq:hsc}
H+H^{\PT}=\Lambda_{\mathcal{E}},
\end{equation}
then $q(H)$ is an upper bound of $p_{\sf PPT}(\mathcal{E})$, that is,
\begin{equation}\label{eq:qpte}
p_{\sf PPT}(\mathcal{E})\leqslant q(H).
\end{equation}
\end{theorem}
%%%%%%%%%%%%%%%%%%%%%%
\begin{proof}
By using the definition of $\Lambda_{\mathcal{E}}$ in Eq.~\eqref{eq:spc} and the notation
\begin{equation}\label{eq:hpm}
H^{(\pm)}=\tfrac{1}{2}(|H|\pm H),
\end{equation}
we can easily see that
\begin{equation}\label{eq:etzo}
(\eta_{0}\rho_{0}-\eta_{1}\rho_{1})-(H^{(+)}-H^{(-)})
-(H^{(+)\PT}-H^{(-)\PT})
=\Lambda_{\mathcal{E}}-H-H^{\PT}=\mathbb{O},
\end{equation}
where the last equality is from Eq.~\eqref{eq:hsc} and $\mathbb{O}$ is the zero operator acting on the Hilbert space $\mathcal{H}=\mathbb{C}^{d_{\sf A}}\otimes\mathbb{C}^{d_{\sf B}}$.
From Eq.~\eqref{eq:etzo}, we have
\begin{equation}\label{eq:hpme}
\eta_{0}\rho_{0}+H^{(-)}+H^{(-)\PT}
=\eta_{1}\rho_{1}+H^{(+)}+H^{(+)\PT}.
\end{equation}
Since $\Tr(AB)=\Tr(A^{\PT}B^{\PT})$ for any Hermitian operators $A$ and $B$, it follows from $H^{(\pm)}\succeq0$ that
\begin{equation}\label{eq:dpdp}
\Tr[(H^{(\pm)}+H^{(\pm)\PT})E]
=\Tr(H^{(\pm)}E)+\Tr(H^{(\pm)}E^{\PT})
\geqslant0
\end{equation}
for any positive-semidefinite PPT operator $E$.
%%%%%%%%%%%%%%%%%%%%%%

%%%%%%%%%%%%%%%%%%%%%%
For any PPT measurement $\mathcal{M}=\{M_{0},M_{1}\}$, we have
\begin{eqnarray}
\sum_{i=0}^{1}\eta_{i}\Tr(\rho_{i}M_{i})
&\leqslant&\Tr(\eta_{0}\rho_{0}M_{0})+\Tr[(H^{(-)}+H^{(-)\PT})M_{0}]
+\Tr(\eta_{1}\rho_{1}M_{1})+\Tr[(H^{(+)}+H^{(+)\PT})M_{1}]\nonumber\\[-1mm]
&=&\Tr[(\eta_{0}\rho_{0}+H^{(-)}+H^{(-)\PT})M_{0}]
+\Tr[(\eta_{1}\rho_{1}+H^{(+)}+H^{(+)\PT})M_{1}]\nonumber\\[1mm]
&=&\tfrac{1}{2}\Tr[(\eta_{0}\rho_{0}+H^{(-)}+H^{(-)\PT})(M_{0}+M_{1})]
+\tfrac{1}{2}\Tr[(\eta_{1}\rho_{1}+H^{(+)}+H^{(+)\PT})(M_{0}+M_{1})]\nonumber\\[1mm]
&=&\tfrac{1}{2}+\Tr H^{(-)}+\Tr H^{(+)}=\tfrac{1}{2}+\Tr |H|=q(H),\label{eq:surs}
\end{eqnarray}
where the first inequality is from Inequality~\eqref{eq:dpdp}, the second equality is due to Eq.~\eqref{eq:hpme}, the third equality is by $M_{0}+M_{1}=\mathbbm{1}$, the fourth equality follows from Eq.~\eqref{eq:hpm}, and the last equality is from the definition of $q(H)$ in Eq.~\eqref{eq:dqe}.
Thus, the definition of $p_{\sf PPT}(\mathcal{E})$ in Eq.~\eqref{eq:pte} and Inequality~\eqref{eq:surs} lead us to Inequality~\eqref{eq:qpte}.
\end{proof}
%%%%%%%%%%%%%%%%%%%%%%
We remark that Eq.~\eqref{eq:hsc} in Theorem~\ref{thm:upbd} is equivalent to 
\begin{equation}\label{eq:lpl}
\Lambda_{\mathcal{E}}^{\PT}=\Lambda_{\mathcal{E}},
\end{equation}
that is, $\Lambda_{\mathcal{E}}$ is invariant under partial transposition.
However, we further remark that the Hermitian operator $H$ satisfying Eq.~\eqref{eq:hsc} is not unique, in general.
The following theorem shows the existence of a Hermitian operator $H$ satisfying Eq.~\eqref{eq:hsc} and $p_{\sf PPT}(\mathcal{E})=q(H)$ when Condition~\eqref{eq:lpl} is satisfied.
%%%%%%%%%%%%%%%%%%%%%%
%      Theorem       %
%%%%%%%%%%%%%%%%%%%%%%
\begin{theorem}\label{thm:qeqp}
For an ensemble $\mathcal{E}=\{\eta_{0},\rho_{0};\eta_{1},\rho_{1}\}$ with Condition~\eqref{eq:lpl}, we have
\begin{equation}\label{eq:qeqp}
p_{\sf PPT}(\mathcal{E})= \min q(H)
\end{equation}
over all possible Hermitian operators $H$ satisfying Eq.~\eqref{eq:hsc}, where $p_{\sf PPT}(\mathcal{E})$ and $q(H)$ are defined in Eqs.~\eqref{eq:pte} and \eqref{eq:dqe}, respectively.
\end{theorem}
%%%%%%%%%%%%%%%%%%%%%%
\begin{proof}
Due to Theorem~\ref{thm:upbd}, it is enough to show the existence of a Hermitian operator $H$ satisfying Eq.~\eqref{eq:hsc} and $p_{\sf PPT}(\mathcal{E})=q(H)$.
We first note that there exists a Hermitian operator $E$ satisfying
\begin{subequations}
\begin{gather}
p_{\sf PPT}(\mathcal{E})=\Tr E,\label{eq:ptte}\\[1mm]
E-\eta_{i}\rho_{i}=P_{i}+Q_{i}^{\PT}\label{eq:eepq}
\end{gather}
\end{subequations}
for all $i=0,1$, where $P_{i}$ and $Q_{i}$ are positive-semidefinite operators\cite{ha20232}.
%%%%%%%%%%%%%%%%%%%%%%

%%%%%%%%%%%%%%%%%%%%%%
For the Hermitian operator
\begin{equation}\label{eq:pqhm}
H=\tfrac{1}{2}(P_{1}+Q_{1}-P_{0}-Q_{0}),
\end{equation}
Eq.~\eqref{eq:hsc} holds because
\begin{eqnarray}
H+H^{\PT}&=&\tfrac{1}{2}(P_{1}+Q_{1}^{\PT}-P_{0}-Q_{0}^{\PT})
+\tfrac{1}{2}(P_{1}+Q_{1}^{\PT}-P_{0}-Q_{0}^{\PT})^{\PT}\nonumber\\[1mm]
&=&\tfrac{1}{2}(\eta_{0}\rho_{0}-\eta_{1}\rho_{1})+\tfrac{1}{2}(\eta_{0}\rho_{0}-\eta_{1}\rho_{1})^{\PT}
=\tfrac{1}{2}\Lambda_{\mathcal{E}}+\tfrac{1}{2}\Lambda_{\mathcal{E}}^{\PT}
=\Lambda_{\mathcal{E}},
\label{eq:hptc}
\end{eqnarray}
where the second equality is by Eq.~\eqref{eq:eepq}, the third equality is from the definition of $\Lambda_{\mathcal{E}}$ in Eq.~\eqref{eq:spc}, and the last equality is due to Condition~\eqref{eq:lpl}. 
Moreover, $p_{\sf PPT}(\mathcal{E})=q(H)$ is satisfied because
\begin{eqnarray}
p_{\sf PPT}(\mathcal{E})&\leqslant&q(H)=\tfrac{1}{2}+\Tr|H|\nonumber\\[1mm]
&\leqslant&\tfrac{1}{2}+\tfrac{1}{2}\Tr(P_{1}+Q_{1})+\tfrac{1}{2}\Tr(P_{0}+Q_{0})\nonumber\\[1mm]
&=&\tfrac{1}{2}+\tfrac{1}{2}\Tr(P_{1}+Q_{1}^{\PT})+\tfrac{1}{2}\Tr(P_{0}+Q_{0}^{\PT})\nonumber\\[1mm]
&=&\Tr E=p_{\sf PPT}(\mathcal{E}),\label{eq:tleq}
\end{eqnarray}
where the first inequality is from Theorem~\ref{thm:upbd}, 
the first equality is due to Eq.~\eqref{eq:dqe}, 
the second inequality is by Eq.~\eqref{eq:pqhm} and the triangle inequality of trace norm, the second equality follows from $\Tr A^{\PT} =\Tr A$ for any Hermitian operator $A$, the third equality is due to Eq.~\eqref{eq:eepq}, and the last equality is from Eq.~\eqref{eq:ptte}.
\end{proof}
%%%%%%%%%%%%%%%%%%%%%%

%%%%%%%%%%%%%%%%%%%%%%
%      Section       %
%%%%%%%%%%%%%%%%%%%%%%
\section{Multi-fold ensemble and hiding one classical bit}\label{sec:mfe}
%%%%%%%%%%%%%%%%%%%%%%
In this section, we first introduce the notion of a \emph{multifold ensemble} and provide a sufficient condition for two-state ensembles to be used in constructing a scheme for hiding one classical bit. In our scheme, the hidden bit can be fully recovered using appropriate global measurements, while LOCC measurements provide arbitrarily little information about the globally accessible bit.
%%%%%%%%%%%%%%%%%%%%%%

%%%%%%%%%%%%%%%%%%%%%%
For a positive integer $L$, we use $\mathbb{Z}_{2}^{L}$ to denote the set of all $L$-dimensional vectors over $\mathbb{Z}_{2}=\{0,1\}$, that is,
\begin{equation}\label{eq:ztl}
\mathbb{Z}_{2}^{L}=\{(b_{1},\ldots,b_{L})\,|\,b_{1},\ldots,b_{L}\in\mathbb{Z}_{2}\}.
\end{equation}
Equivalently, we can also consider $\mathbb{Z}_{2}^{L}$ as the set of all $L$-bit strings.
For the ensemble $\mathcal{E}$ in Eq.~\eqref{eq:ens} and each string $\vec{b}\in\mathbb{Z}_{2}^{L}$, we define
\begin{equation}\label{eq:ervb}
\eta_{\vec{b}}=\prod_{l=1}^{L}\eta_{b_{l}},~~\rho_{\vec{b}}=\bigotimes_{l=1}^{L}\rho_{b_{l}},
\end{equation}
where $b_{l}$ is the $l$th entry of $\vec{b}=(b_{1},\ldots,b_{L})$. 
We further denote by $\mathcal{E}^{\otimes L}$ the \emph{$L$-fold ensemble} of $\mathcal{E}$, that is,
\begin{equation}\label{eq:eol}
\mathcal{E}^{\otimes L}=\{\eta_{\vec{b}},\rho_{\vec{b}}\}_{\vec{b}\in\mathbb{Z}_{2}^{L}},
\end{equation}
in which the state $\rho_{\vec{b}}$ is prepared with the probability $\eta_{\vec{b}}$.
For each string $\vec{b}\in\mathbb{Z}_{2}^{L}$, we define $\omega_{2}(\vec{b})$ as the modulo 2 summation of all the entries in $\vec{b}=(b_{1},\ldots,b_{L})$, that is,
\begin{equation}\label{eq:mta}
\omega_{2}(\vec{b})=b_{1}\oplus\cdots\oplus b_{L}\in\mathbb{Z}_{2},
\end{equation}
where $\oplus$ denotes the modulo 2 addition in $\mathbb{Z}_{2}$.
%%%%%%%%%%%%%%%%%%%%%%

%%%%%%%%%%%%%%%%%%%%%%
Let us consider the situation of guessing $\omega_{2}(\vec{b})$ of the prepared state $\rho_{\vec{b}}$ from the ensemble $\mathcal{E}^{\otimes L}$ in Eq.~\eqref{eq:eol}.
This situation is equivalent to discriminating the two bipartite quantum states $\rho_{0}^{(L)}$ and $\rho_{1}^{(L)}$ prepared with the probabilities $\eta_{0}^{(L)}$ and $\eta_{1}^{(L)}$, respectively, where
\begin{equation}\label{eq:eril}
\eta_{i}^{(L)}=\sum_{\substack{\vec{b}\in\mathbb{Z}_{2}^{L}\\ \omega_{2}(\vec{b})=i}}\eta_{\vec{b}},~~
\rho_{i}^{(L)}=\frac{1}{\eta_{i}^{(L)}}\sum_{\substack{\vec{b}\in\mathbb{Z}_{2}^{L}\\ \omega_{2}(\vec{b})=i}}\eta_{\vec{b}}\rho_{\vec{b}},~~i=0,1.
\end{equation}
In other words, quantum state discrimination of the ensemble
\begin{equation}\label{eq:esl}
\mathcal{E}^{(L)}=\{\eta_{0}^{(L)},\rho_{0}^{(L)};\eta_{1}^{(L)},\rho_{1}^{(L)}\}.
\end{equation}
We note that
\begin{equation}\label{eq:lael}
\Lambda_{\mathcal{E}^{(L)}}=\eta_{0}^{(L)}\rho_{0}^{(L)}-\eta_{1}^{(L)}\rho_{1}^{(L)}
=(\eta_{0}\rho_{0}-\eta_{1}\rho_{1})^{\otimes L}=\Lambda_{\mathcal{E}}^{\otimes L},
\end{equation}
where $\Lambda_{\mathcal{E}}$ is defined in Eq.~\eqref{eq:spc} and the second equality 
is satisfied because Eq.~\eqref{eq:eril} implies
\begin{equation}\label{eq:erle}
\eta_{i}^{(L)}\rho_{i}^{(L)}=\frac{1}{2}\Big[(\eta_{0}\rho_{0}+\eta_{1}\rho_{1})^{\otimes L}
+(-1)^{i}(\eta_{0}\rho_{0}-\eta_{1}\rho_{1})^{\otimes L}\Big],~i=0,1.
\end{equation}
%%%%%%%%%%%%%%%%%%%%%%

%%%%%%%%%%%%%%%%%%%%%%
Now, let us consider a bipartite \emph{orthogonal} quantum state ensemble $\mathcal{E}=\{\eta_{0},\rho_{0};\eta_{1},\rho_{1}\}$ satisfying 
\begin{equation}\label{eq:lptu}
\lim_{L\rightarrow\infty}p_{\sf L}(\mathcal{E}^{(L)})=\frac{1}{2}.
\end{equation}
This ensemble $\mathcal{E}$ can be used to construct a quantum data-hiding scheme that hides one classical bit\cite{ha20241}. 
The scheme consists of two steps:
%%%%%%%%%%%%%%%%%%%%%%

%%%%%%%%%%%%%%%%%%%%%%
Step 1. The hider, the extra party, first choose a state $\rho_{b_{1}}$ from $\mathcal{E}$ with the corresponding probability $\eta_{b_{1}}$ for $b_{1}\in\mathbb{Z}_{2}$ and share it with Alice and Bob. 
The hider repeats this process $L$ times, so that $\rho_{b_{l}}$ is chosen with probability $\eta_{b_{l}}$ in the $l$th repetition for $l=1,\ldots,L$. 
This whole procedure is equivalent to the situation in which the hider first prepares the state
\begin{equation}\label{eq:psrb}
\rho_{\vec{b}}=\rho_{b_{1}}\otimes\cdots\otimes\rho_{b_{L}}
\end{equation}
from the ensemble $\mathcal{E}^{\otimes L}$ in Eq.~\eqref{eq:eol} with the corresponding probability $\eta_{\vec{b}}$ for $\vec{b}=(b_{1},\ldots,b_{L})\in\mathbb{Z}_{2}^{L}$ and sends it to Alice and Bob.
%%%%%%%%%%%%%%%%%%%%%%

%%%%%%%%%%%%%%%%%%%%%%
Step 2. To hide one classical bit $x\in\mathbb{Z}_{2}$, the hider broadcasts the classical information $z$ to Alice and Bob,
\begin{equation}\label{eq:zxy}
z=x\oplus y
\end{equation}
where 
\begin{equation}\label{eq:yob}
y=\omega_{2}(\vec{b})
\end{equation}
for the prepared state $\rho_{\vec{b}}$ in Eq.~\eqref{eq:psrb}.
%%%%%%%%%%%%%%%%%%%%%%

%%%%%%%%%%%%%%%%%%%%%%
As the classical information $z$ in Eq.~\eqref{eq:zxy} was broadcasted, guessing the data $x$ is equivalent to guessing the data $y$ in Eq.~\eqref{eq:yob}.
Since the situation of guessing the data $y$ is equivalent to the state discrimination of $\mathcal{E}^{(L)}$, the maximum average probability of correctly guessing the data $x$ is equal to the maximum average probability of discriminating the states from the ensemble $\mathcal{E}^{(L)}$ in Eq.~\eqref{eq:esl}, that is, $p_{\sf G}(\mathcal{E}^{(L)})$.
Similarly, the maximum average probability of correctly guessing the data $x$ using only LOCC measurements becomes $p_{\sf L}(\mathcal{E}^{(L)})$.
%%%%%%%%%%%%%%%%%%%%%%

%%%%%%%%%%%%%%%%%%%%%%
The orthogonality between the states of $\mathcal{E}$ implies that the states of $\mathcal{E}^{\otimes L}$ in Eq.~\eqref{eq:eol} are mutually orthogonal.
In this case, the states $\rho_{0}^{(L)}$ and $\rho_{1}^{(L)}$ in Eq.~\eqref{eq:eril} are orthogonal. 
Thus, we have
\begin{equation}\label{eq:pgeo}
p_{\sf G}(\mathcal{E}^{(L)})=1.
\end{equation}
In other words, the data $x$ are accessible when Alice and Bob can collaborate to use global measurements.
On the other hand, Eq.~\eqref{eq:lptu} means $p_{\sf L}(\mathcal{E}^{(L)})$ can be arbitrarily close to $1/2$ with the choice of an appropriately large $L$. Thus, the globally accessible data $x$ are perfectly hidden asymptotically if Alice and Bob can use only LOCC measurements.
%%%%%%%%%%%%%%%%%%%%%%

%%%%%%%%%%%%%%%%%%%%%%
From Inequalities~\eqref{eq:pgplr} and \eqref{eq:lptm}, we have
\begin{equation}\label{eq:hppt}
\tfrac{1}{2}\leqslant p_{\sf L}(\mathcal{E}^{(L)})\leqslant p_{\sf PPT}(\mathcal{E}^{(L)})
\end{equation}
for each positive integer $L$. 
To show Eq.~\eqref{eq:lptu}, it is enough to show that $p_{\sf PPT}(\mathcal{E}^{(L)})$ converges to $1/2$ as $L$ gets larger.
The following theorem provides a sufficient condition for $\{p_{\sf PPT}(\mathcal{E}^{(L)})\}_{L}$ to be a monotonically decreasing sequence with respect to $L$.
As $\{p_{\sf PPT}(\mathcal{E}^{(L)})\}_{L}$ is bounded below by $1/2$, the monotonicity of the sequence assures that $\{p_{\sf PPT}(\mathcal{E}^{(L)})\}_{L}$ is a convergent sequence.
%%%%%%%%%%%%%%%%%%%%%%
%      Theorem       %
%%%%%%%%%%%%%%%%%%%%%%
\begin{theorem}\label{thm:perl}
For an ensemble $\mathcal{E}=\{\eta_{0},\rho_{0};\eta_{1},\rho_{1}\}$ with Condition~\eqref{eq:lpl}, $p_{\sf PPT}(\mathcal{E}^{(L)})$ is monotonically decreasing as $L$ increases.
\end{theorem}
%%%%%%%%%%%%%%%%%%%%%%
\begin{proof}
For an arbitrary positive integer $L$, we have
\begin{equation}\label{eq:hhpc}
\Lambda_{\mathcal{E}^{(L)}}^{\PT}
=(\Lambda_{\mathcal{E}}^{\otimes L})^{\PT}
=(\Lambda_{\mathcal{E}}^{\otimes L})
=\Lambda_{\mathcal{E}^{(L)}},
\end{equation}
where the first and last equalities are from Eq.~\eqref{eq:lael} and the second equality is due to Condition~\eqref{eq:lpl}.
From Theorem~\ref{thm:qeqp}, there exists a Hermitian operator $H$ satisfying
\begin{subequations}\label{eq:hcpc0}
\begin{gather}
H+H^{\PT}=\Lambda_{\mathcal{E}^{(L)}},\label{eq:hcpc1}\\[1mm]
p_{\sf PPT}(\mathcal{E}^{(L)})=q(H).\label{eq:hcpc2}
\end{gather}
\end{subequations}
%%%%%%%%%%%%%%%%%%%%%%

%%%%%%%%%%%%%%%%%%%%%%
For the Hermitian operator
\begin{equation}\label{eq:hphc}
H'=\frac{1}{2}H\otimes\Lambda_{\mathcal{E}}+\frac{1}{2}\Lambda_{\mathcal{E}}\otimes H,
\end{equation}
we have
\begin{eqnarray}
H'+H'^{\PT}&=&\tfrac{1}{2}(H+H^{\PT})\otimes\Lambda_{\mathcal{E}}
+\tfrac{1}{2}\Lambda_{\mathcal{E}}\otimes(H+H^{\PT})\nonumber\\[1mm]
&=&\tfrac{1}{2}\Lambda_{\mathcal{E}^{(L)}}\otimes\Lambda_{\mathcal{E}}+\tfrac{1}{2}\Lambda_{\mathcal{E}}\otimes\Lambda_{\mathcal{E}^{(L)}}
=\Lambda_{\mathcal{E}^{(L+1)}},\label{eq:hphpt}
\end{eqnarray}
where the first equality is due to Condition~\eqref{eq:lpl}, the second equality is from Eq.~\eqref{eq:hcpc1}, and the last equality is by Eq.~\eqref{eq:lael}.
Moreover, we have
\begin{eqnarray}
p_{\sf PPT}(\mathcal{E}^{(L+1)})&\leqslant&q(H')=\tfrac{1}{2}+\Tr|H'|\nonumber\\[1mm]
&=&\tfrac{1}{2}+\Tr|\tfrac{1}{2}H\otimes\Lambda_{\mathcal{E}}+\tfrac{1}{2}\Lambda_{\mathcal{E}}\otimes H|\nonumber\\[1mm]
&\leqslant&\tfrac{1}{2}+\Tr|\tfrac{1}{2}H\otimes\Lambda_{\mathcal{E}}|+\Tr|\tfrac{1}{2}\Lambda_{\mathcal{E}}\otimes H|\nonumber\\[1mm]
&=&\tfrac{1}{2}+\Tr|H|\Tr|\Lambda_{\mathcal{E}}|\leqslant\tfrac{1}{2}+\Tr|H|\nonumber\\[1mm]
&=&q(H)=p_{\sf PPT}(\mathcal{E}^{(L)}),\label{eq:thlc}
\end{eqnarray}
where the first inequality is from Eq.~\eqref{eq:hphpt} and Theorem~\ref{thm:upbd}, the second inequality is by the triangle inequality of trace norm, the third equality is due to $\Tr|A\otimes B|=\Tr|A|\Tr|B|$ for any Hermitian operators $A$ and $B$, the last inequality is from the definition of $\Lambda_{\mathcal{E}}$ in Eq.~\eqref{eq:spc} together with the triangle inequality of trace norm, and the last equality is due to Eq.~\eqref{eq:hcpc2}.
Since the choice of $L$ can be arbitrary, $p_{\sf PPT}(\mathcal{E}^{(L)})$ is monotonically decreasing as a function of $L$.
\end{proof}
%%%%%%%%%%%%%%%%%%%%%%
%     Corollary      %
%%%%%%%%%%%%%%%%%%%%%%
\begin{corollary}\label{cor:pell}
For an ensemble $\mathcal{E}=\{\eta_{0},\rho_{0};\eta_{1},\rho_{1}\}$ with Condition~\eqref{eq:lpl}, if a subsequence of $\{p_{\sf PPT}(\mathcal{E}^{(L)})\}_{L}$ converges to $1/2$, then Eq.~\eqref{eq:lptu} is satisfied.
\end{corollary}
%%%%%%%%%%%%%%%%%%%%%%

%%%%%%%%%%%%%%%%%%%%%%
%      Section       %
%%%%%%%%%%%%%%%%%%%%%%
\section{Quantum data-hiding scheme using separable state ensemble}\label{sec:ex}
%%%%%%%%%%%%%%%%%%%%%%
In this section, we use Corollary~\ref{cor:pell} to provide a quantum data-hiding scheme using separable state ensemble.
To obtain a sufficient condition for the convergence of $\{p_{\sf PPT}(\mathcal{E}^{(L)})\}_{L}$ to $1/2$, we consider the subsequence $\{p_{\sf PPT}(\mathcal{E}^{(3^{m})})\}_{m}$ of $\{p_{\sf PPT}(\mathcal{E}^{(L)})\}_{L}$, and obtain a condition for the convergence of $\{p_{\sf PPT}(\mathcal{E}^{(3^{m})})\}_{m}$ to $1/2$.
More precisely, we provide a sufficient condition for an upper bound of $p_{\sf PPT}(\mathcal{E}^{(3^{m})})$ to converge to $1/2$ as $m$ gets larger.
%%%%%%%%%%%%%%%%%%%%%%

%%%%%%%%%%%%%%%%%%%%%%
The following lemma provides a method to obtain an upper bound of $p_{\sf PPT}(\mathcal{E}^{(3)})$ from a Hermitian operator $H$ satisfying Eq.~\eqref{eq:hsc}.
%%%%%%%%%%%%%%%%%%%%%%
%       Lemma        %
%%%%%%%%%%%%%%%%%%%%%%
\begin{lemma}\label{lem:hufg}
For an ensemble $\mathcal{E}=\{\eta_{0},\rho_{0};\eta_{1},\rho_{1}\}$, if there exists a Hermitian operator $H$ satisfying Eq.~\eqref{eq:hsc} and 
\begin{equation}\label{eq:fhpo}
\Tr|H|+\Tr|H^{\PT}|\leqslant 1,
\end{equation}
then there exists a Hermitian operator $H'$ satisfying
\begin{subequations}
\begin{gather}
H'+H'^{\PT}=\Lambda_{\mathcal{E}}^{\otimes 3},\label{eq:hplt}\\[1mm]
\Tr|H'|\leqslant f(\Tr|H|),\label{eq:thf1}\\[1mm]
\Tr|H'|+\Tr|H'^{\PT}|\leqslant 1\label{eq:fhf}
\end{gather}
\end{subequations}
where $f$ is the polynomial function defined as
\begin{equation}\label{eq:dfg}
f(x)=x^{2}(3-2x).
\end{equation}
\end{lemma}
%%%%%%%%%%%%%%%%%%%%%%
\noindent Before we prove Lemma~\ref{lem:hufg}, we first remark that Lemma~\ref{lem:hufg} directly leads us to an upper bound of $p_{\sf PPT}(\mathcal{E}^{(3)})$ as
\begin{equation}\label{eq:upbw}
p_{\sf PPT}(\mathcal{E}^{(3)})\leqslant
q(H')\leqslant \tfrac{1}{2}+ f(\Tr|H|),
\end{equation}
where the first inequality is from Theorem~\ref{thm:upbd} and Eq.~\eqref{eq:lael}, and the last inequality is due to Inequalities~\eqref{eq:thf1}.
%%%%%%%%%%%%%%%%%%%%%%
\begin{proof}
Let us consider
\begin{equation}\label{eq:hhhp}
H'=H\otimes H\otimes H+H^{\PT}\otimes H\otimes H
+H\otimes H^{\PT}\otimes H+H\otimes H\otimes H^{\PT}.
\end{equation}
The Hermitian operator $H'$ in Eq.~\eqref{eq:hhhp} satisfies Eq.~\eqref{eq:hplt} because
\begin{equation}\label{eq:hhht}
H'+H'^{\PT}=(H+H^{\PT})\otimes(H+H^{\PT})
\otimes(H+H^{\PT})
=\Lambda_{\mathcal{E}}^{\otimes3},
\end{equation}
where the last equality is due to Eq.~\eqref{eq:hsc}.
%%%%%%%%%%%%%%%%%%%%%%

%%%%%%%%%%%%%%%%%%%%%%
Inequality~\eqref{eq:thf1} is satisfied because
\begin{eqnarray}\label{eq:tnhp}
\Tr|H'|&=&\Tr|\,(H+H^{\PT})\otimes H\otimes H
+H\otimes H^{\PT}\otimes H+H\otimes H\otimes H^{\PT}\,|\nonumber\\[1mm]
&=&\Tr|\,\Lambda_{\mathcal{E}}\otimes H\otimes H +H\otimes H^{\PT}\otimes H
+H\otimes H\otimes H^{\PT}\,|\nonumber\\[1mm]
&\leqslant&\Tr|\,\Lambda_{\mathcal{E}}\otimes H\otimes H\,|
+\Tr|\,H\otimes H^{\PT}\otimes H\,|
+\Tr|\,H\otimes H\otimes H^{\PT}\,|\nonumber\\[1mm]
&=&(\Tr|\Lambda_{\mathcal{E}}|)(\Tr|H|)^{2}+2(\Tr|H|)^{2}\Tr|H^{\PT}|\nonumber\\[1mm]
&\leqslant&(\Tr|H|)^{2}(1+2\Tr|H^{\PT}|)\nonumber\\[1mm]
&\leqslant&(\Tr|H|)^{2}(3-2\Tr|H|)\nonumber\\[1mm]
&=&f(\Tr|H|),
\end{eqnarray}
where the second equality is due to Eq.~\eqref{eq:hsc}, the first inequality is from the triangle inequality of trace norm, the third equality is by $\Tr|A\otimes B|=\Tr|A|\Tr|B|$ for any Hermitian operators $A$ and $B$, the second inequality is by the definition of $\Lambda_{\mathcal{E}}$ in Eq.~\eqref{eq:spc} together with the triangle inequality of trace norm, the last inequality is due to Eq.~\eqref{eq:fhpo}, and the last equality follows from the definition of $f$ in Eq.~\eqref{eq:dfg}.
Analogously, we also have
\begin{equation}\label{eq:awcs}
\Tr|H'^{\PT}|\leqslant f(\Tr|H^{\PT}|).
\end{equation}
%%%%%%%%%%%%%%%%%%%%%%

%%%%%%%%%%%%%%%%%%%%%%
Moreover, Inequality~\eqref{eq:fhf} holds because
\begin{equation}\label{eq:thth}
\Tr|H'|+\Tr|H'^{\PT}|
\leqslant f(\Tr|H|)+f(\Tr|H^{\PT}|)\\
\leqslant f(\Tr|H|)+f(1-\Tr|H|)=1,
\end{equation}
where the first inequality is by Inequalities~\eqref{eq:tnhp} and \eqref{eq:awcs}, the last inequality is from Inequality~\eqref{eq:fhpo} together with $f(x)<f(y)$ for any real numbers $x,y$ with $0\leqslant x<y\leqslant 1$, and the equality is due to $f(x)+f(1-x)=1$ for any real number $x$.
\end{proof}
%%%%%%%%%%%%%%%%%%%%%%

%%%%%%%%%%%%%%%%%%%%%%
Now, we iteratively use Lemma~\ref{lem:hufg} to obtain the following theorem which provides an upper bound of $p_{\sf PPT}(\mathcal{E}^{(3^{m})})$ for each positive integer $m$.
%%%%%%%%%%%%%%%%%%%%%%
%      Theorem       %
%%%%%%%%%%%%%%%%%%%%%%
\begin{theorem}\label{thm:hptm}
For an ensemble $\mathcal{E}=\{\eta_{0},\rho_{0};\eta_{1},\rho_{1}\}$, if there exist a Hermitian operator $H$ satisfying Eq.~\eqref{eq:hsc} and Inequality~\eqref{eq:fhpo}, then 
\begin{equation}\label{eq:upuu}
p_{\sf PPT}(\mathcal{E}^{(3^{m})})\leqslant\tfrac{1}{2}+f^{[m]}(\Tr|H|)
\end{equation}
for any non-negative integer $m$, where $f^{[m]}$ is the $m$th iterate of the polynomial function $f$ defined in Eq.~\eqref{eq:dfg}.
\end{theorem}
%%%%%%%%%%%%%%%%%%%%%%
\begin{proof}
For any non-negative integer $m$, we show that there exists a Hermitian operator $H'$ satisfying
\begin{subequations}\label{eq:flis}
\begin{gather}
H'+H'^{\PT}=\Lambda_{\mathcal{E}}^{\otimes 3^{m}},\label{eq:fltn}\\[1mm]
\Tr|H'|\leqslant f^{[m]}(\Tr|H|),\label{eq:fiht}\\[1mm]
\Tr|H'|+\Tr|H'^{\PT}|\leqslant 1,\label{eq:hpph}
\end{gather}
\end{subequations}
therefore Theorem~\ref{thm:upbd} and Eq.~\eqref{eq:lael} lead us to Inequality~\eqref{eq:upuu}.
%%%%%%%%%%%%%%%%%%%%%%

%%%%%%%%%%%%%%%%%%%%%%
For the case when $m=0$, the Hermitian operator $H'=H$ clearly satisfies Conditions~\eqref{eq:flis}.
For the case when $m>0$, we use the mathematical induction on $m$; we first assume that the theorem is true for any non-negative integer less than or equal to $m-1$, that is, the existence of a Hermitian operator $H^{\circ}$ satisfying
\begin{subequations}\label{eq:ftcc}
\begin{gather}
H^{\circ}+H^{\circ\PT}=\Lambda_{\mathcal{E}}^{\otimes3^{m-1}},\label{eq:fchc}\\[1mm]
\Tr|H^{\circ}|\leqslant f^{[m-1]}(\Tr|H|),\label{eq:fctc}\\[1mm]
\Tr|H^{\circ}|+\Tr|H^{\circ\PT}|\leqslant 1,\label{eq:hpf}
\end{gather}
\end{subequations}
and show the validity of Conditions~\eqref{eq:flis}.
%%%%%%%%%%%%%%%%%%%%%%

%%%%%%%%%%%%%%%%%%%%%%
Since $\Lambda_{\mathcal{E}}^{\otimes 3^{m-1}}=\Lambda_{\mathcal{E}^{(3^{m-1})}}$ due to Eq.~\eqref{eq:lael}, it follows from Lemma~\ref{lem:hufg} and the induction hypothesis in \eqref{eq:ftcc} that there exists a Hermitian operator $H'$ satisfying Condition~\eqref{eq:hpph} and
\begin{subequations}\label{eq:fcca}
\begin{gather}
H'+H'^{\PT}=\Lambda_{\mathcal{E}^{(3^{m-1})}}^{\otimes 3},\label{eq:fhca}\\[1mm]
\Tr|H'|\leqslant f(\Tr|H^{\circ}|).\label{eq:fclt}
\end{gather}
\end{subequations}
Equation~\eqref{eq:fhca} means Condition~\eqref{eq:fltn} because
\begin{equation}\label{eq:lmlm}
\Lambda_{\mathcal{E}^{(3^{m-1})}}^{\otimes 3}=(\Lambda_{\mathcal{E}}^{\otimes 3^{m-1}})^{\otimes 3}=\Lambda_{\mathcal{E}}^{\otimes 3^{m}},
\end{equation}
where the first equality is from Eq.~\eqref{eq:lael}.
Since $0\leqslant f(x)< f(y)\leqslant 1$ for any real numbers $x,y$ with $0\leqslant x<y\leqslant 1$, we can see from Inequalities~\eqref{eq:fhpo} and \eqref{eq:fctc} that
\begin{equation}\label{eq:fthc}
f(\Tr|H^{\circ}|)\leqslant f(f^{[m-1]}(\Tr|H|))=f^{[m]}(\Tr|H|).
\end{equation}
Inequalities~\eqref{eq:fclt} and \eqref{eq:fthc} lead us to Condition~\eqref{eq:fiht}.
Thus, our theorem is true.
\end{proof}
%%%%%%%%%%%%%%%%%%%%%%

%%%%%%%%%%%%%%%%%%%%%%
By using Theorem~\ref{thm:hptm}, we provides a sufficient condition for the convergence of $\{p_{\sf PPT}(\mathcal{E}^{(3^{m})})\}_{m}$ to $1/2$.
%%%%%%%%%%%%%%%%%%%%%%
%      Theorem       %
%%%%%%%%%%%%%%%%%%%%%%
\begin{theorem}\label{thm:ubu}
For an ensemble $\mathcal{E}=\{\eta_{0},\rho_{0};\eta_{1},\rho_{1}\}$, if there exists a Hermitian operator $H$ satisfying Eq.~\eqref{eq:hsc} and Inequality~\eqref{eq:fhpo} together with 
\begin{equation}\label{eq:thfc}
\Tr|H|<\tfrac{1}{2}, 
\end{equation}
then $p_{\sf PPT}(\mathcal{E}^{(3^{m})})$ converges to $1/2$ exponentially fast with respect to $m$.
\end{theorem}
%%%%%%%%%%%%%%%%%%%%%%
\begin{proof}
Due to Theorem~\ref{thm:hptm} and Inequality~\eqref{eq:upuu}, it is enough to show that $f^{[m]}(\Tr|H|)$ converges to $0$ exponentially fast with respect to $m$.
When $\Tr|H|=0$, our statement is clearly true due to $f(0)=0$.
%%%%%%%%%%%%%%%%%%%%%%

%%%%%%%%%%%%%%%%%%%%%%
Let us suppose $\Tr|H|>0$.
Since the function $f$ satisfies $0<f(x)/x<f(y)/y<1$ for any real numbers $x,y$ with $0< x<y<1/2$, we can see from Inequalities~\eqref{eq:fhpo} and \eqref{eq:thfc} that
\begin{equation}\label{eq:uuk}
0<\frac{f(f^{[m]}(\Tr|H|))}{f^{[m]}(\Tr|H|)}<\frac{f(f^{[m-1]}(\Tr|H|))}{f^{[m-1]}(\Tr|H|)}<1
\end{equation}
for any positive integer $m$.
From Inequality~\eqref{eq:uuk}, we have
\begin{equation}\label{eq:umir}
f^{[m]}(\Tr|H|)=\Tr|H|\prod_{k=0}^{m-1}\frac{f(f^{[k]}(\Tr|H|))}{f^{[k]}(\Tr|H|)}
<\Tr|H|\Bigg[\frac{f(\Tr|H|)}{\Tr|H|}\Bigg]^{m}
\end{equation}
for any positive integer $m$.
Thus, $f^{[m]}(\Tr|H|)$ converges to $0$ exponentially fast with respect to $m$.
\end{proof}
%%%%%%%%%%%%%%%%%%%%%%

%%%%%%%%%%%%%%%%%%%%%%
\begin{example}\label{ex:tts}
Let us consider the $3\otimes3$ orthogonal separable state ensemble $\mathcal{E}=\{\eta_{0},\rho_{0};\eta_{1},\rho_{1}\}$ consisting of
\begin{align}\label{eq:tte}
\eta_{0}=\frac{1}{2},~
\rho_{0}=&\frac{1}{4}\sum_{i=0}^{3}\ket{\alpha_{i}}\!\bra{\alpha_{i}}\otimes\ket{\alpha_{i}}\!\bra{\alpha_{i}},\nonumber\\
\eta_{1}=\frac{1}{2},~
\rho_{1}=&\frac{1}{6}\sum_{i=0}^{2}\big(\ket{\beta_{i}^{+}}\!\bra{\beta_{i}^{+}}\otimes\ket{\beta_{i}^{-}}\!\bra{\beta_{i}^{-}}
+\ket{\beta_{i}^{-}}\!\bra{\beta_{i}^{-}}\otimes\ket{\beta_{i}^{+}}\!\bra{\beta_{i}^{+}}\big)
\end{align}
where
\begin{equation}\label{eq:fso}
\begin{array}{ll}
\ket{\alpha_{0}}=\frac{1}{\sqrt{3}}(\ket{0}+\ket{1}+\ket{2}),
&\ket{\beta_{0}^{\pm}}=\frac{1}{\sqrt{2}}(\ket{0}\pm\ket{1}),\\[2mm]
\ket{\alpha_{1}}=\frac{1}{\sqrt{3}}(\ket{0}-\ket{1}-\ket{2}),
&\ket{\beta_{1}^{\pm}}=\frac{1}{\sqrt{2}}(\ket{1}\pm\ket{2}),\\[2mm]
\ket{\alpha_{2}}=\frac{1}{\sqrt{3}}(\ket{0}-\ket{1}+\ket{2}),
&\ket{\beta_{2}^{\pm}}=\frac{1}{\sqrt{2}}(\ket{2}\pm\ket{0}),\\[2mm]
\ket{\alpha_{3}}=\frac{1}{\sqrt{3}}(\ket{0}+\ket{1}-\ket{2}).&
\end{array}
\end{equation}
Note that the states $\rho_{0}$ and $\rho_{1}$ are orthogonal because either $\braket{\alpha_{i}|\beta_{j}^{+}}$ or $\braket{\alpha_{i}|\beta_{j}^{-}}$ is zero for all $i=0,1,2,3$ and all $j=0,1,2$.
\end{example}
%%%%%%%%%%%%%%%%%%%%%%
For the Hermitian operator
\begin{equation}\label{eq:heo}
H=\frac{1}{6}\ket{\Gamma_{0}}\!\bra{\Gamma_{0}}
-\frac{5}{48}\sum_{i=1}^{2}\ket{\Gamma_{i}}\!\bra{\Gamma_{i}}
+\frac{1}{72}\sum_{i=0}^{2}\ket{\Gamma_{i}^{+}}\!\bra{\Gamma_{i}^{+}}
\end{equation}
where $\{\ket{\Gamma_{i}},\ket{\Gamma_{i}^{+}},\ket{\Gamma_{i}^{-}}\}_{i=0}^{2}$ is an orthonormal basis of $\mathbb{C}^{3}\otimes\mathbb{C}^{3}$,
\begin{equation}\label{eq:xis}
\begin{array}{ll}
\ket{\Gamma_{0}}=\frac{1}{\sqrt{3}}(\ket{00}+\ket{11}+\ket{22}),
&
\ket{\Gamma_{0}^{\pm}}=\frac{1}{\sqrt{2}}(\ket{12}\pm\ket{21}),
\\[2mm]
\ket{\Gamma_{1}}=\frac{1}{\sqrt{3}}(\ket{00}+e^{\frac{\mathrm{i}2\pi}{3}}\ket{11}+e^{\frac{\mathrm{i}4\pi}{3}}\ket{22}),
&
\ket{\Gamma_{1}^{\pm}}=\frac{1}{\sqrt{2}}(\ket{20}\pm\ket{02}),
\\[2mm]
\ket{\Gamma_{2}}=\frac{1}{\sqrt{3}}(\ket{00}+e^{\frac{\mathrm{i}4\pi}{3}}\ket{11}+e^{\frac{\mathrm{i}2\pi}{3}}\ket{22}),
&
\ket{\Gamma_{2}^{\pm}}=\frac{1}{\sqrt{2}}(\ket{01}\pm\ket{10}),
\end{array}
\end{equation}
it is straightforward to check Eq.~\eqref{eq:hsc} and its partial transposition
\begin{equation}\label{eq:pthe}
H^{\PT}=
-\frac{1}{48}\sum_{i=1}^{2}\ket{\Gamma_{i}}\!\bra{\Gamma_{i}}
+\sum_{i=0}^{2}\Big(\frac{7}{72}\ket{\Gamma_{i}^{+}}\!\bra{\Gamma_{i}^{+}}
-\frac{1}{12}\ket{\Gamma_{i}^{-}}\!\bra{\Gamma_{i}^{-}}\Big).
\end{equation}
%%%%%%%%%%%%%%%%%%%%%%

%%%%%%%%%%%%%%%%%%%%%%
From Eqs.~\eqref{eq:heo} and \eqref{eq:pthe}, we can easily see that
\begin{equation}\label{eq:tfst}
\Tr|H|=\frac{5}{12},~~\Tr|H^{\PT}|=\frac{7}{12}.
\end{equation}
Since Eq.~\eqref{eq:tfst} and Theorem~\ref{thm:ubu} lead us to the convergence of $\{p_{\sf PPT}(\mathcal{E}^{(3^{m})})\}_{m}$ to $1/2$, it follows from Corollary~\ref{cor:pell} that Eq.~\eqref{eq:lptu} holds.
Thus, the orthogonal separable state ensemble $\mathcal{E}$ in Example~\ref{ex:tts} can be used to construct a data-hiding scheme that hides one classical bit.
The details of a quantum data-hiding scheme are given in Sec.~\ref{sec:mfe}.
%%%%%%%%%%%%%%%%%%%%%%

%%%%%%%%%%%%%%%%%%%%%%
Example~\ref{ex:tts} can be generalized into higher-dimensional bipartite quantum systems.
For an odd number $d\geqslant3$, the following example provides a two-qu$d$it orthogonal separable state ensemble that can be used to construct a quantum data-hiding scheme.
%%%%%%%%%%%%%%%%%%%%%%
\begin{example}\label{ex:ddex}
For an odd number $d\geqslant3$, let us consider the $d\otimes d$ orthogonal separable state ensemble $\mathcal{E}=\{\eta_{0},\rho_{0};\eta_{1},\rho_{1}\}$ consisting of
\begin{align}
\eta_{0}=\frac{1}{2},~
\rho_{0}=&\frac{1}{2^{d-1}}
\hskip -5mm
\sum_{\substack{S\subseteq\{0,1,\ldots,d-1\}\\ |S|<d/2}}
\hskip -5mm
\ket{\phi_{S}}\!\bra{\phi_{S}}\otimes\ket{\phi_{S}}\!\bra{\phi_{S}},\nonumber\\
\eta_{1}=\frac{1}{2},~
\rho_{1}=&\frac{1}{d(d-1)}\sum_{\substack{i,j=0\\i<j}}^{d-1}\big(\ket{\psi_{ij}^{+}}\!\bra{\psi_{ij}^{+}}\otimes\ket{\psi_{ij}^{-}}\!\bra{\psi_{ij}^{-}}
+\ket{\psi_{ij}^{-}}\!\bra{\psi_{ij}^{-}}\otimes\ket{\psi_{ij}^{+}}\!\bra{\psi_{ij}^{+}}\big)\label{eq:dde}
\end{align}
where
\begin{equation}\label{eq:ddf}
\ket{\phi_{S}}=\frac{1}{\sqrt{d}}\sum_{i=0}^{d-1}\ket{i}
-\frac{2}{\sqrt{d}}\sum_{j\in S}\ket{j},~~
\ket{\psi_{ij}^{\pm}}=\frac{1}{\sqrt{2}}(\ket{i}\pm\ket{j}).
\end{equation}
Note that the states $\rho_{0}$ and $\rho_{1}$ are orthogonal since either $\braket{\phi_{S}|\psi_{ij}^{+}}$ or $\braket{\phi_{S}|\psi_{ij}^{-}}$ is zero for all $S\subseteq\{0,1,\ldots,d-1\}$ with $|S|<d/2$ and all $i,j=0,1,\ldots,d-1$ with $i<j$.
\end{example}
%%%%%%%%%%%%%%%%%%%%%%

%%%%%%%%%%%%%%%%%%%%%%
Before we provide a Hermitian operator $H$ satisfying Eq.~\eqref{eq:hsc}, we show that
\begin{subequations}\label{eq:rrt}
\begin{eqnarray}
\eta_{0}\rho_{0}&=&\frac{1}{2d}\Phi
+\frac{1}{d^{2}}\Psi^{+},\label{eq:rsrw}\\
\eta_{1}\rho_{1}&=&\frac{1}{4(d-1)}(\Pi-\Phi)
+\frac{1}{2d(d-1)}\Psi^{-}\label{eq:rsrt}
\end{eqnarray}
\end{subequations}
where 
\begin{equation}\label{eq:dxis}
\Pi=\sum_{i=0}^{d-1}\ket{ii}\!\bra{ii},~~
\Phi=\frac{1}{d}\sum_{i,j=0}^{d-1}\ket{ii}\!\bra{jj},~~
\Psi^{\pm}=\frac{1}{2}\sum_{\substack{i,j=0\\i<j}}^{d-1}(\ket{ij}\pm\ket{ji})(\bra{ij}\pm\bra{ji}).
\end{equation}
We can show Eq.~\eqref{eq:rsrt} by tedious but straightforward calculations.
On the other hand, showing Eq.~\eqref{eq:rsrw} requires a tricky calculation using binomial coefficients
\begin{equation}\label{eq:cnk}
\binom{n}{k}=\frac{n!}{k!(n-k)!}.
\end{equation}
%%%%%%%%%%%%%%%%%%%%%%

%%%%%%%%%%%%%%%%%%%%%%
From the definitions of $\eta_{0}$ and $\rho_{0}$ in Eq.~\eqref{eq:dde}, we can see that
\begin{eqnarray}\label{eq:rosr}
\eta_{0}\rho_{0}
&=&\frac{1}{2^{d}}\sum_{i,j,i',j'=0}^{d-1}\ket{ij}\!\bra{i'j'}
\sum_{\substack{S\subseteq\{0,1,\ldots,d-1\}\\ |S|<d/2}}
\braket{i|\phi_{S}}\!\braket{j|\phi_{S}}\braket{\phi_{S}|i'}\!\braket{\phi_{S}|j'}
\nonumber\\
&=&\frac{1}{2^{d}d^{2}}\sum_{i,j,i',j'=0}^{d-1}\ket{ij}\!\bra{i'j'}
\sum_{\substack{S\subseteq\{0,1,\ldots,d-1\}\\ |S|<d/2}}
\mu_{S}(i)\mu_{S}(j)\mu_{S}(i')\mu_{S}(j')
\nonumber\\
&=&\frac{1}{2^{d}d^{2}}\sum_{i=0}^{d-1}\Bigg[\ket{ii}\!\bra{ii}
+\sum_{\substack{j=0\\ j\neq i}}^{d-1}(\ket{ii}\!\bra{jj}
+\ket{ij}\!\bra{ij}+\ket{ij}\!\bra{ji})
\Bigg]
\sum_{\substack{S\subseteq\{0,1,\ldots,d-1\}\\ |S|<d/2}}
1
\nonumber\\
&&+\frac{1}{2^{d}d^{2}}\sum_{\substack{i,j=0\\ i\neq j}}^{d-1}\Bigg[\ket{ij}\!\bra{jj}
+\ket{ji}\!\bra{jj}+\ket{jj}\!\bra{ij}+\ket{jj}\!\bra{ji}
+\sum_{\substack{i'=0\\ i'\neq i,j}}^{d-1}
(\ket{i'i'}\!\bra{ij}+\ket{i'i}\!\bra{i'j}
\nonumber\\
&&+\ket{i'i}\!\bra{ji'}
+\ket{ii'}\!\bra{i'j}+\ket{ii'}\!\bra{ji'}+\ket{ij}\!\bra{i'i'})\Bigg]
\sum_{\substack{S\subseteq\{0,1,\ldots,d-1\}\\ |S|<d/2}}
\mu_{S}(i)\mu_{S}(j)
\nonumber\\
&&+\frac{1}{2^{d}d^{2}}
\sum_{\substack{i,j=0\\ i\neq j}}^{d-1}
\sum_{\substack{i'=0\\ i'\neq i,j}}^{d-1}
\sum_{\substack{j'=0\\ j'\neq i,j,i'}}^{d-1}
\ket{ij}\!\bra{i'j'}
\sum_{\substack{S\subseteq\{0,1,\ldots,d-1\}\\ |S|<d/2}}
\mu_{S}(i)\mu_{S}(j)\mu_{S}(i')\mu_{S}(j')
\end{eqnarray}
where
\begin{equation}\label{eq:mudf}
\mu_{S}(i)=\left\{
\begin{array}{rl}
-1,&i\in S,\\[1mm]
1,&i\notin S.
\end{array}
\right.
\end{equation}
The last equality in Eq.~\eqref{eq:rosr} is due to $\mu_{S}^{2}(i)=1$ for any element $i$ and subset $S$ of $\{0,1,\ldots,d-1\}$.
%%%%%%%%%%%%%%%%%%%%%%

%%%%%%%%%%%%%%%%%%%%%%
Since the greatest integer less than $d/2$ is $(d-1)/2$ by oddness of $d$, it follows that
\begin{equation}\label{eq:ssdk}
\sum_{\substack{S\subseteq\{0,1,\ldots,d-1\}\\ |S|<d/2}}1
=\sum_{k=0}^{(d-1)/2}\binom{d}{k}
=\frac{1}{2}\sum_{k=0}^{d}\binom{d}{k}
=2^{d-1},
\end{equation}
where the second equality is by $\binom{d}{k}=\binom{d}{d-k}$ for any integer $k$ with $0\leqslant k\leqslant d$ and
the last equality is by $\sum_{k=0}^{n}\binom{d}{k}=2^{d}$.
For any $i,j\in\{0,1,\ldots,d-1\}$ with $i\neq j$, we have
\begin{eqnarray}\label{eq:msms}
\sum_{\substack{S\subseteq\{0,1,\ldots,d-1\}\\ |S|<d/2}}
\mu_{S}(i)\mu_{S}(j)
&=&\sum_{\substack{S\subseteq\{0,1,\ldots,d-1\}\\ |S|<d/2,~i,j\in S}}{\hskip -3mm} 1
+\sum_{\substack{S\subseteq\{0,1,\ldots,d-1\}\\ |S|<d/2,~i,j\notin S}}{\hskip -3mm}1
-2{\hskip -3mm}\sum_{\substack{S\subseteq\{0,1,\ldots,d-1\}\\ |S|<d/2,~i\in S,~j\notin S}}{\hskip -3mm}1\nonumber\\[2mm]
&=&\sum_{k=0}^{(d-5)/2}\binom{d-2}{k}
+\sum_{k=0}^{(d-1)/2}\binom{d-2}{k}
-2\sum_{k=0}^{(d-3)/2}\binom{d-2}{k}
=0,
\end{eqnarray}
where the first equality is from the definition of $\mu_{S}$ in Eq.~\eqref{eq:mudf} and
the last equality is by $\binom{d-2}{k}=\binom{d-2}{d-2-k}$ for any integer $k$ with $0\leqslant k\leqslant d-2$.
%%%%%%%%%%%%%%%%%%%%%%

%%%%%%%%%%%%%%%%%%%%%%
For any distinct $i,j,i',j'\in\{0,1,\ldots,d-1\}$, we also have
\begin{eqnarray}\label{eq:mmmm}
\sum_{\substack{S\subseteq\{0,1,\ldots,d-1\}\\ |S|<d/2}}
\mu_{S}(i)\mu_{S}(j)\mu_{S}(i')\mu_{S}(j')
&=&
6{\hskip -3mm}\sum_{\substack{S\subseteq\{0,1,\ldots,d-1\}\\ |S|<d/2,\\ i,j\in S,~ i',j'\notin S}}{\hskip -3mm}1
+\sum_{\substack{S\subseteq\{0,1,\ldots,d-1\}\\ |S|<d/2,\\ i,j,i',j'\in S}}{\hskip -3mm}1
+\sum_{\substack{S\subseteq\{0,1,\ldots,d-1\}\\ |S|<d/2,\\i,j,i',j'\notin S}}{\hskip -3mm}1
\nonumber\\
&&-4{\hskip -3mm}\sum_{\substack{S\subseteq\{0,1,\ldots,d-1\}\\ |S|<d/2,~i\in S\\ j,i',j'\notin S}}{\hskip -3mm}1
-4{\hskip -3mm}\sum_{\substack{S\subseteq\{0,1,\ldots,d-1\}\\ |S|<d/2,~i\notin S\\ j,i',j'\in S}}{\hskip -3mm}1
\nonumber\\[2mm]
&=&6\sum_{k=0}^{(d-5)/2}\binom{d-4}{k}+\sum_{k=0}^{(d-9)/2}\binom{d-4}{k}
+\sum_{k=0}^{(d-1)/2}\binom{d-4}{k}
\nonumber\\
&&-4\sum_{k=0}^{(d-3)/2}\binom{d-4}{k}
-4\sum_{k=0}^{(d-7)/2}\binom{d-4}{k}
=0,
\end{eqnarray}
where the first equality is from the definition of $\mu_{S}$ in Eq.~\eqref{eq:mudf} and the last equality is by $\binom{d-4}{k}=\binom{d-4}{d-4-k}$ for any integer $k$ with $0\leqslant k\leqslant d-4$.
Applying Eqs.~\eqref{eq:ssdk}, \eqref{eq:msms} and \eqref{eq:mmmm} to Eq.~\eqref{eq:rosr}, 
we can easily verify Eq.~\eqref{eq:rsrw}.
%%%%%%%%%%%%%%%%%%%%%%

%%%%%%%%%%%%%%%%%%%%%%
Now, we consider the Hermitian operator
\begin{equation}\label{eq:dheo}
H=\frac{1}{2d}\Phi
-\frac{d+2}{8d(d-1)}(\Pi-\Phi)
+\frac{d-2}{4d^{2}(d-1)}\Psi^{+}.
\end{equation}
Since the definitions of $\Pi$, $\Phi$, and $\Psi^{\pm}$ in Eq.~\eqref{eq:dxis} imply
\begin{equation}\label{eq:vvpt}
\Pi^{\PT}=\Pi,~~
\Phi^{\PT}=\frac{1}{d}(\Pi+\Psi^{+}-\Psi^{-}),~~
\Psi^{\pm\PT}=\pm \frac{d}{2}\Phi\mp\frac{1}{2}\Pi
+\frac{1}{2}(\Psi^{+}+\Psi^{-}),
\end{equation}
we have
\begin{equation}\label{eq:dthe}
H^{\PT}=\frac{d-2}{8d(d-1)}(\Phi-\Pi)
+\frac{3d-2}{4d^{2}(d-1)}\Psi^{+}-\frac{1}{2d(d-1)}\Psi^{-}.
\end{equation}
We can easily check  Eq.~\eqref{eq:hsc} using Eq.~\eqref{eq:rrt} together with Eqs.~\eqref{eq:dheo} and \eqref{eq:dthe}.
%%%%%%%%%%%%%%%%%%%%%%

%%%%%%%%%%%%%%%%%%%%%%
From the spectral decompositions of $H$ and $H^{\PT}$ in Eqs.~\eqref{eq:dheo} and \eqref{eq:dthe}, we have
\begin{equation}\label{eq:ddst}
\Tr|H|=\frac{d+2}{4d},~~\Tr|H^{\PT}|=\frac{3d-2}{4d}.
\end{equation}
Since Eq.~\eqref{eq:ddst} and Theorem~\ref{thm:ubu} lead us to the convergence of $\{p_{\sf PPT}(\mathcal{E}^{(3^{m})})\}_{m}$ to $1/2$, it follows from Corollary~\ref{cor:pell} that Eq.~\eqref{eq:lptu} holds.
Thus, the orthogonal separable state ensemble $\mathcal{E}$ in Example~\ref{ex:ddex} can be used to construct a scheme that hides one classical bit.
The details of a quantum data-hiding scheme are given in Sec.~\ref{sec:mfe}.
%%%%%%%%%%%%%%%%%%%%%%

%%%%%%%%%%%%%%%%%%%%%%
%      Section       %
%%%%%%%%%%%%%%%%%%%%%%
\section{Discussion}\label{sec:dis}
We have considered the discrimination of two bipartite quantum states and provide a quantum data-hiding scheme using orthogonal separable states. We have established a bound on the local minimum-error discrimination for bipartite state ensembles with specific properties related to partial transposition. Based on this bound, we have introduced a quantum data-hiding scheme that utilizes an orthogonal state ensemble exhibiting these same properties. We have further established a sufficient condition for an ensemble of two bipartite quantum states to be used in constructing a one-bit hiding scheme. 
Finally, our results have been illustrated with examples of orthogonal separable state ensembles satisfying this condition. 
%%%%%%%%%%%%%%%%%%%%%%

%%%%%%%%%%%%%%%%%%%%%%
Our data-hiding scheme ensures perfect data recovery without relying on asymptotic methods due to the orthogonality of the states. Furthermore, classical data is asymptotically concealed from LOCC measurements by repeatedly employing separable states of low-dimensional quantum systems. This makes our scheme more feasible for practical implementation especially in noisy intermediate-scale quantum environments. 
%%%%%%%%%%%%%%%%%%%%%%

%%%%%%%%%%%%%%%%%%%%%%
An interesting direction for future research is to explore separable-state data-hiding schemes hiding more than single bit or schemes for multi-users in multipartite quantum systems. It is also an interesting and important direction for future research to investigate which types of quantum correlations are necessary for quantum data hiding.  
%%%%%%%%%%%%%%%%%%%%%%

%%%%%%%%%%%%%%%%%%%%%%
%  Acknowledgments   %
%%%%%%%%%%%%%%%%%%%%%%
\section*{Acknowledgments}
This work was supported by a National Research Foundation of Korea(NRF) grant funded by the Korean government(Ministry of Science and ICT)(No.NRF2023R1A2C1007039). JSK was supported by Creation of the Quantum Information Science R\&D Ecosystem (Grant No. 2022M3H3A106307411) through the National Research Foundation of Korea (NRF) funded by the Korean government (Ministry of Science and ICT).
%%%%%%%%%%%%%%%%%%%%%%

%%%%%%%%%%%%%%%%%%%%%%
%     Appendix       %
%%%%%%%%%%%%%%%%%%%%%%
%\appendix
%%%%%%%%%%%%%%%%%%%%%%
%     References     %
%%%%%%%%%%%%%%%%%%%%%%

%%%%%%%%%%%%%%%%%%%%%%

\begin{thebibliography}{99}
\bibitem{pptp}  
The eigenvalues of a partially transposed operator do not depend on the choice of basis or the subsystem to be transposed. 
For simplicity, we consider the standard basis and the second system throughout this paper.
\end{thebibliography}
\end{document}